\documentclass[showpacs,preprintnumbers,twocolumn]{revtex4}

\usepackage[centertags]{amsmath}
\usepackage{amsfonts}
\usepackage{amssymb}
\usepackage{amsthm}
\usepackage{newlfont}
\usepackage{graphicx}
\input{epsfx}

\newtheorem{thm}{Theorem}[section]
\newtheorem{lem}[thm]{Lemma}

\newtheorem{defn}[thm]{Definition}
\newtheorem{cor}[thm]{Corollary}

\newtheorem{prop}[thm]{Proposition}

\begin{document}

\title{Monogamy of quantum entanglement in time}
\author{Marcin Nowakowski\footnote{Electronic address: mnowakowski@mif.pg.gda.pl}}
\affiliation{Faculty of Applied Physics and Mathematics,
~Gdansk University of Technology, 80-952 Gdansk, Poland}
\affiliation{National Quantum Information Center of Gdansk, Andersa 27, 81-824 Sopot, Poland}


\begin{abstract}
In this paper we state a fundamental question about the structure of correlations in time and analyze temporal monogamy relations.
We show that the nature of temporal correlations is inherently different from the spatial ones but in similarity to quantum spatial correlations, we expose a phenomenon of monogamy of quantum entanglement in time. We perform this task applying the entangled histories framework as a modification of the consistent histories approach. These considerations are supported by introduction of necessary tools specific for the tensor algebra used for representation of
spatial correlations. We show that Tsirelson bound on temporal Bell-like inequalities can be derived from the entangled histories approach.  Finally, we point out that in a context of the tensor algebra used for linking states in different times further studies on mathematical structure of the state representing evolving systems are needed.
\end{abstract}

\maketitle

\section{Introduction}

Recent years have proved a great interest of quantum entanglement monogamy concept showing its usability in quantum communication theory, especially in domain of one-way communication and its applications to quantum secure key generation \cite{Lutk1,Lutk2,Lutk3,Lutk4,Devetak05,KLi,MNPH,MNPH2}. While spatial quantum correlations and especially their non-locality became a central subject of quantum information theory and their applications to quantum computation, potentiality of application of temporal non-local correlations is poorly analyzed. The crucial issue relates to the very nature of time and temporal correlations phenomenon with their understanding within the framework of modern quantum and relativistic theories.

Non-local nature of quantum correlations in space has been accepted as a consequence of violation of local realism, expressed in Bell's theorem \cite{Bell} and analyzed in many experiments \cite{Aspect, Freedman}. As an analogy for a temporal domain, violation of macro-realism \cite{LGI2} and Legett-Garg inequalities \cite{LGI} seem to indicate non-local effects in time and are a subject of many experimental considerations \cite{EX1, EX2, EX3, EX4}. However, the open problem relates to the mathematical structures that could represent quantum states correlated in time in similarity to multipartite quantum states in space. In this paper we analyze a variation of the consistent histories approach \cite{RG1, RG2, RG3, RG4} with a concept of entangled histories \cite{CJ1,CJ2} built on a tensor product of projective Hilbert spaces that can be considered as a potential candidate of mathematical structures representing quantum states correlated in time. In particular, we focus on showing that entangled histories demonstrate monogamous properties reflecting the phenomenon in case of spatial quantum entanglement.  It is worth mentioning that the two-state-vector formalism (TSVF) \cite{Vaidman} brings another perspective on representation of quantum correlations in time broadly discussed in the literature.

However, it is crucial to note that in this context many 'obvious' facts about structure and behavior of spatial correlations and tensor algebra of spatial quantum states cannot be easily transferred into the temporal domain as the tensor structure of temporal correlations is richer due to the binding evolution between instances of 'time' and the observation-measurement phenomenon that is also a subject of this paper.

The outline of this paper is as follows: in section I, we present the well known concepts of consistent histories approach \cite{RG3} and present new concepts of entangled histories \cite{WC1, WC2} which are substantial for further considerations on monogamies and entanglement in time as such. In section II, we introduce partial trace on quantum histories and show that quantum entanglement in time is monogamous for a particular history. This section considers also this property from a perspective of the Feynman's path integral approach. In section III, the Tsirelson bound on quantum correlations in time is derived from the entangled histories.

We believe that further research on temporal correlations and time evolution will be substantial for development of quantum information theory including applications to quantum cryptography or quantum computation but also to quantum gravity theory.

\section{Entangled Consistent Histories}

The consistent histories approach has a long tradition \cite{RG1,RG2, RG3,RG4,RG5, RG6, CJ1, CJ2} and as such resolves many quantum 'paradoxes' but is also a subject of many open discussions.
For readers interested in deepening the subject of consistent histories, it might be useful to refer to the literature \cite{RG3}. In this section we give a short introduction with necessary tools for further reading of the paper but
also introduce a new concept of monogamy of correlations in time from the consistent histories perspective. Partial trace operator on $\mathcal{C}^{*}$-Algebra of history operators is introduced as a tool necessary for analysis of reduced histories justifying its consistency with the Feynman's path integral approach \cite{Feynman}.

It is substantial to notice that for a temporally evolving system we can ask questions about its states probing the system in different
times $t_{1}<t_{2}<...<t_{n}$ that is performed in reference to the measuring device. We could interpret that during this process we project the state of the system onto the n-fold tensor product $\bigodot_{i=1}^{n}P_{i}$ achieving
a consistent wave function which can be used to deduce probabilities of the events where a history of a physical system is defined by means of the product $\bigodot_{i=1}^{n}P_{i}$ of local projectors $P_{i}$ acting on the Hilbert space
$\mathcal{H}=\mathcal{H}_{t_{n}}\otimes\dots\otimes\mathcal{H}_{t_{1}}$ in which
the particular states live. The projectors in a given history represent then potential properties of the system it had in analyzed times. It means that we build a tensor algebra from temporally local frames (identified by pointer $t_{i}$) analyzing the global state of evolving system in local Hilbert spaces $H_{i}$. Thus, one can interpret a history $Proj(\mathcal{H})\ni |H)=P_{n}^{1}\odot P_{n-1}^{1}\odot\ldots\odot P_{1}^{1}$ as a potential evolution where the system has a property $P_{i}$ in time $t_{i}$ \cite{RG3}. To have a physical sense, it means that in principle it should be possible to define an 'observable' for such a history state that is defined as  $\mathcal{O}=|H)(H|$ \cite{WC1,WC2}.
Naturally, it is still an open question what mathematical structure would be sufficient to describe 'atemporal state' of the system
in any reference frame with which one can associate all possible physical history observables.

The fundamental tool introduced in the consistent history framework which connects different frames is the bridging operator \cite{RG1} $\mathcal{B}(t_{2},t_{1})$. It is a counterpart of an unitary evolution operation having the following properties:
\begin{eqnarray}
\mathcal{B}(t_{2},t_{1})^{\dag}&=&\mathcal{B}(t_{1},t_{2}) \\
\mathcal{B}(t_{3},t_{2})\mathcal{B}(t_{2},t_{1})&=&\mathcal{B}(t_{3},t_{1})
\end{eqnarray}
and can be represented for unitary quantum evolution as $\mathcal{B}(t_{2},t_{1})=\exp(-iH(t_{2}-t_{1}))$.

Now for the sample space of consistent histories $|H^{\alpha})=P_{n}^{1}\odot P_{n-1}^{\alpha}\odot\ldots\odot P_{1}^{\alpha}\odot P_{0}^{\alpha}$ where $\sum_{\alpha}|H^{\alpha})=id$, the formalism introduces the chain operator $K(|H^{\alpha}))$ which will be fundamental for answering the questions about probabilities of realization of such a history that could be assigned to a particular history via the Born rule:
\begin{equation}
K(|H^{\alpha}))=P_{n}^{\alpha}\mathcal{B}(t_{n},t_{n-1}) P_{n-1}^{\alpha}\ldots\mathcal{B}(t_{2},t_{1})P_{1}^{\alpha}\mathcal{B}(t_{1},t_{0}) P_{0}^{\alpha}
\end{equation}
Further, equipped with this operator one can associate a history $|H^{\alpha})$ with its weight $W(|H^{\alpha}))=TrK(|H^{\alpha}))^{\dagger}K(|H^{\alpha}))$ being by Born rule a counterpart of relative probabilities.
The histories framework requires additionally that the family of histories is consistent, i.e. one can associate with a union of histories a weight equal to the sum of weights
associated with particular histories included in the union. This implies the following \textit{consistency condition}:
\begin{equation}
 \left\lbrace
           \begin{array}{l}
(H^{\alpha}|H^{\beta})\equiv TrK(|H^{\alpha}))^{\dagger}K(|H^{\beta}))=0 \; for\; \alpha\neq\beta \\
 (H^{\alpha}|H^{\beta})=0 \; or \; 1 \\
 \sum_{\alpha}c_{\alpha}|H^{\alpha})=I\; for\; complex\; c_{\alpha} \in \mathcal{C} \\
\end{array}
         \right.
\end{equation}
If the observed system starts its potential history in a pure state
$P_{t_{0}}=|\Psi_{0}\rangle\langle\Psi_{0}|$, then a consistent set of its histories create a tree-like structure (Fig. 1). The consistency condition implies that the tree branches are mutually orthogonal. It is helpful also to assume normalization of histories with non-zero weight which enables normalization of probability distributions for history events, i.e.: $|\widetilde{H})=\frac{|H)}{\sqrt{(H|H)}}$.

The consistent history framework does not consider non-locality in space or time as such \cite{RG7}, however, since the space of histories spans the complex vector space, we can consider
complex combinations of history vectors, i.e. any history can be represented as $|\Psi)=\sum_{i}\alpha_{i}|H^{i})$ \cite{WC1} where $\alpha_{i} \in \mathbb{C}$ and $\mathcal{F} \ni |H^{i})$ represents a consistent family of histories which is actually a complex extension of the framework built by R. Griffiths.
Having defined above, the histories space can be also equipped with an inner semi-definite product \cite{RG1} between any two histories $|\Psi)$ and $|\Phi)$ as $(\Psi|\Phi)=Tr [K(|\Psi))^{\dagger}K(|\Phi))]$.
It is worth mentioning that recently \cite{WC1, WC2} the concept of Bell-like tests have been proposed for experimental analysis of entangled histories.

It is fundamental to note that a history $|H^{\alpha})$ can be consistent or inconsistent (physically not realizable) basing on the associated evolution $T$ of the system \cite{RG3} as its consistency is verified by means
of the aforementioned inner product engaging bridging operators. Thus, a temporal history is always associated with evolution and one actually always should consider a pair of a family of histories and the bridging operators  $\{\mathcal{F}, T\}$ for the physical system. This is the first key feature differentiating vectors representing spatial quantum states (as objects from Hilbert or Banach space for mixed states) and temporal vectors representing histories.
Whenever we analyze features of a spatial quantum state, it is assumed that all necessary knowledge is hidden in the vector $|\psi\rangle$ so actually we analyze only one-element history objects $[\psi]=|\psi\rangle\langle\psi|$
from a perspective of a temporal local frame without a concept of time.

\section{Monogamy of Quantum Entanglement in Time}

Now we are ready to put a fundamental question about monogamy of correlations inheriting from the entangled histories approach. Since the algebra with $\bigodot$ operation is a form of tensor algebra, it inherits all properties
 of normal tensor algebra and all mathematical questions valid for vectors representing spatial correlations are mathematically valid for temporal correlations although not necessarily having similar
  physical interpretation.
It is substantial to note that any history $Y=F_{n}\odot\ldots\odot F_{0}$ can be extended to $I\odot Y$ as identity I represents a property that is always true and does not introduce additional
knowledge about the system.
Conversely, if one considers reduction of a history to smaller number of time frames, then information about the past and future of the reduced history is lost. Let us consider the potential history of the
physical system $Y_{t_{n}\ldots t_{0}}=F_{n}\odot F_{n-1}\odot\ldots\odot F_{2}\odot F_{1}\odot F_{0}$ on times $\{t_{n}\ldots t_{0}\}$, then at time $t_{1}$ the reduced history is $Y_{t_{1}}=F_{1}$.
Further, looking at the history $Y_{t_{n}\ldots t_{0}}$ one can associate with $Y_{t_{1}}$ two bridging operators $\mathcal{B}(t_{2},t_{1})^{\dag}$ and $\mathcal{B}(t_{1},t_{0})$
by means of which we can calculate two propagators taking the history from the future and past events to $F_{1}$ where a system evolves through the potential paths consistent with this history.


The aforementioned considerations are vital for further statements about potential nature of monogamic relations in time for a particular entangled history of a physical system.
We can introduce now a partial trace operation on a history similarly to the spatial operator in accordance with a general rule of calculating partial traces on tensor algebras:
\begin{defn}\label{PartialTrace}
For a history $|Y_{t_{n}\ldots t_{0}})=F_{n}\odot F_{n-1}\odot\ldots\odot F_{1}\odot F_{0}$ acting on a space $\mathcal{H}=\mathcal{H}_{t_{n}}\otimes\dots\otimes\mathcal{H}_{t_{0}}$, a partial trace over times $\{t_{j}\ldots t_{i+1} t_{i}\}$ $(j\geq i)$ is:
\[
Tr_{t_{j}\ldots t_{i+1}t_{i}} |Y_{t_{n}\ldots t_{0}})(Y_{t_{n}\ldots t_{0}}|=\sum_{k=1}^{\dim\mathcal{F}} (e_{k}|Y_{t_{n}\ldots t_{0}})(Y_{t_{n}\ldots t_{0}}|e_{k})
\]
where $\mathcal{F}=\{|e_{k})\}$ creates an orthonormal consistent family of histories on times $\{t_{j}\ldots t_{i+1} t_{i}\}$ and the strong consistency condition for partial histories holds for base histories, i.e. $(e_{i}|e_{j})=Tr[K(|e_{i}))^{\dag}K(|e_{j}))]=\delta_{ij}$.

\end{defn}

\begin{figure}
\includegraphics[width=7cm, height=5cm]{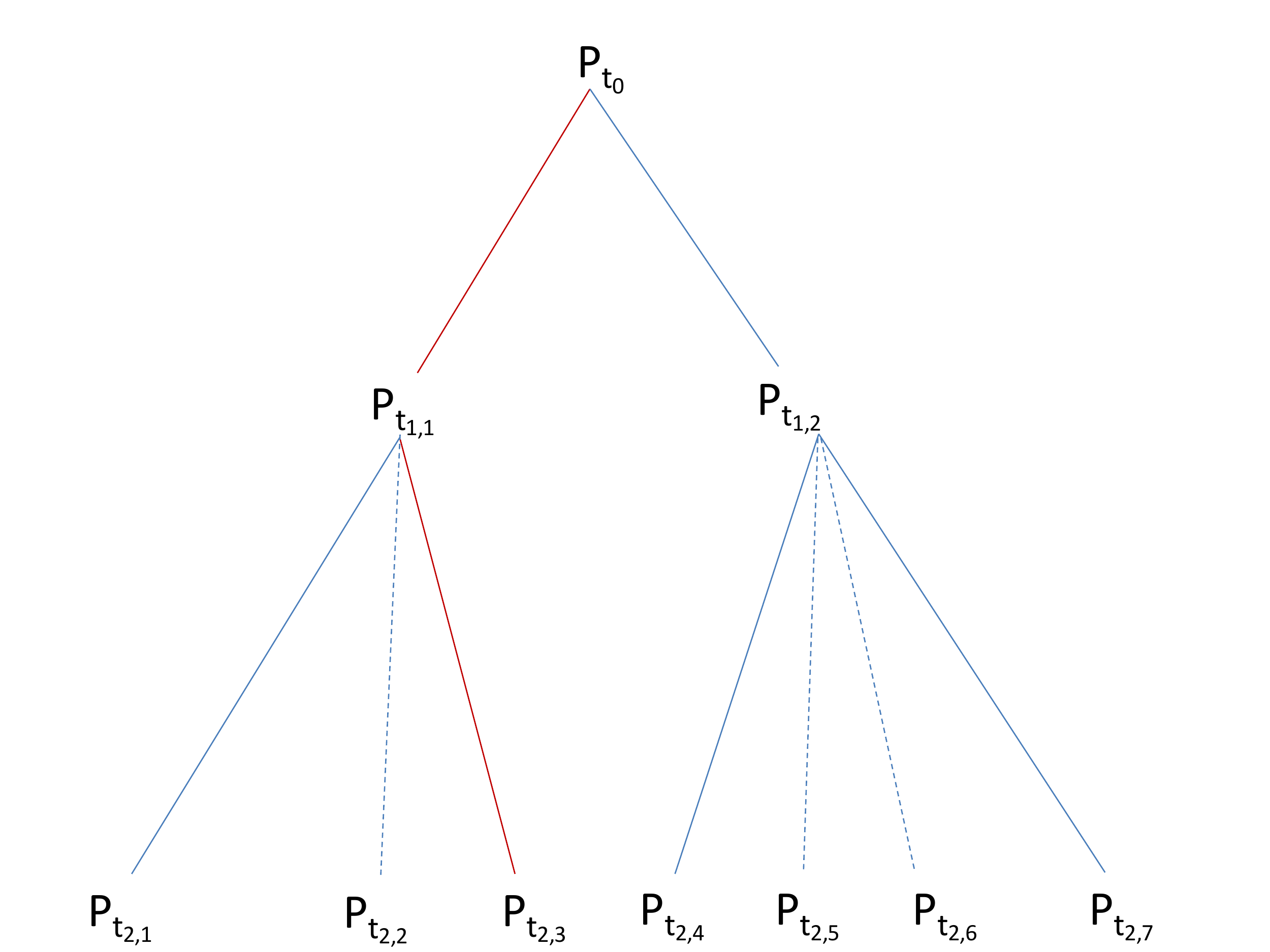}
 \caption[Fig1.]{If the observed evolution is initiated in a state $[P_{t_{0}}]=|\Psi_{0}\rangle\langle\Psi_{0}|$, then the history family can be represented as a tree-like structure where each brach represents a potential history. The branches are mutually orthogonal due to the consistency condition. The exemplary red branch represents history $H=P_{t_{2,3}}\odot P_{t_{1,1}}\odot P_{t_{0}}$.}
 \label{Fig 1}
\end{figure}

Basing on the concept of entangled histories, we propose further a general form of maximally entangled history in similarity to maximally entangled bipartite qudit state in space $|\Psi_{+}\rangle=\frac{1}{\sqrt{N}}\sum_{i=1}^{N}|i\rangle\otimes|i\rangle,\; 2\leq N<\infty$:

\begin{prop}\label{maxent}
A state 'maximally entangled' in time would be represented by:
\begin{equation}
|\Psi)=\frac{1}{\sqrt{N}}\sum_{i=1}^{N}|e_{i})\odot|e_{i}),\; 2\leq N<\infty
\end{equation}
with a trivial bridging operator $I$ and  $\{|e_{i})\}$ creating an orthonormal consistent histories family.
\end{prop}

It is important to note that one can always employ such a bridging operator that $|\Psi)$ could become intrinsically inconsistent which means it would be dynamically impossible \cite{RG3}, thus, an identity bridging operator is associated with the above state.
Further, one could also introduce $\tau GHZ$ and $\tau W$ states substantial for studies of multipartite correlations and their applications (e.g. for secret key generation, quantum algorithms or spin networks) in analogy to spatial $|GHZ\rangle$ and $|W\rangle$ states with trivial bridging operators:
\begin{equation}
 \left\lbrace
           \begin{array}{l}
|\tau GHZ)=  \frac{1}{\sqrt{2}}( |e_{0})^{\odot N}+|e_{1})^{\odot N})\\
|\tau W)= \frac{1}{\sqrt{N}}( |e_{1})\odot|e_{0})\odot\cdots \odot|e_{0}) +\\|e_{0})\odot|e_{1})\odot\cdots \odot|e_{0})+\cdots+|e_{0})\odot|e_{0})\odot\cdots \odot|e_{1}))\\
\end{array}
         \right.
\end{equation}

To better understand the behavior of entangled histories in a context of monogamy, we further consider the Mach-Zender interferometer (Fig.\ref{Fig 2}).

\textit{Example 1.}
Namely, let us consider the following potential intrinsically consistent history on times $\{t_{3}, t_{2}, t_{1}, t_{0}\}$:
\begin{equation}
|\Lambda)=\alpha([\phi_{3,1}]\odot I_{t_{2}}\odot [\phi_{1,1}]+[\phi_{3,2}]\odot I_{t_{2}}\odot [\phi_{1,2}])\odot [\phi_{0}]
\end{equation}
where $\alpha$ stands for the normalization factor, $[\phi_{i,j}]=|\phi_{i,j}\rangle\langle\phi_{i,j}|$ and potentiality of the history means that one can construct a history observable $\widehat{\Lambda}=|\Lambda)(\Lambda|$.
Now, retracting only the state on times $t_{1}$ and $t_{3}$ one gets the reduced history:
\begin{equation}
|\Lambda_{1})=\tilde{\alpha}([\phi_{3,1}]\odot [\phi_{1,1}]+[\phi_{3,2}]\odot [\phi_{1,2}])
\end{equation}
which displays entanglement in time apparently.
Noticeably, we have to show that to be in agreement with the partial trace definition and Feynman propagators' formalism \cite{Feynman} the history $|\Lambda_{1})$ cannot be extracted from
the following $|\tau GHZ)$-like state which is also allowed in the setup of the aforementioned interferometer (Fig. 2) as a potential history:
\begin{equation}
|\Psi)=\gamma([\phi_{3,1}]\odot [\phi_{2,1}]\odot [\phi_{1,1}]+[\phi_{3,2}]\odot [\phi_{2,2}]\odot [\phi_{1,2}])
\end{equation}

We observe that the reduced history $[\phi_{3,1}]\odot [\phi_{1,1}]$ is correlated with $[\phi_{2,1}]$ and not with $[\phi_{2,2}]$. Thus, we cannot simply add the histories $[\phi_{3,1}]\odot [\phi_{1,1}]+[\phi_{3,2}]\odot [\phi_{1,2}]$ as a reduction. It would imply decorrelation with the next instance of the history in such a case, i.e. it could be always expanded to a history e.g. $[\phi_{t_{x}}]\odot([\phi_{3,1}]\odot [\phi_{1,1}]+[\phi_{3,2}]\odot [\phi_{1,2}])$. The latter is
in agreement with the Feynman's addition rule for probability amplitudes as that would mean existence of detectors in the consecutive step. It is important
to note that these considerations are related to $|\Psi)(\Psi|$ - observable and the reference frame associated with that. It shows clearly a physical sense
of quantum entanglement in time and further a concept of its monogamy for a particular entangled history.

One finds temporal monogamy phenomenon in similarity to the spatial monogamy \cite{Wootters} on the ground of consistent histories approach saying that we cannot build a tripartite state $\rho_{ABC}$ where
$\rho_{AB}=\rho_{BC}=|\Psi )( \Psi|$ and $Tr_{C}\rho_{ABC}=\rho_{AB}$.

\begin{figure}
\includegraphics[width=7cm, height=5cm]{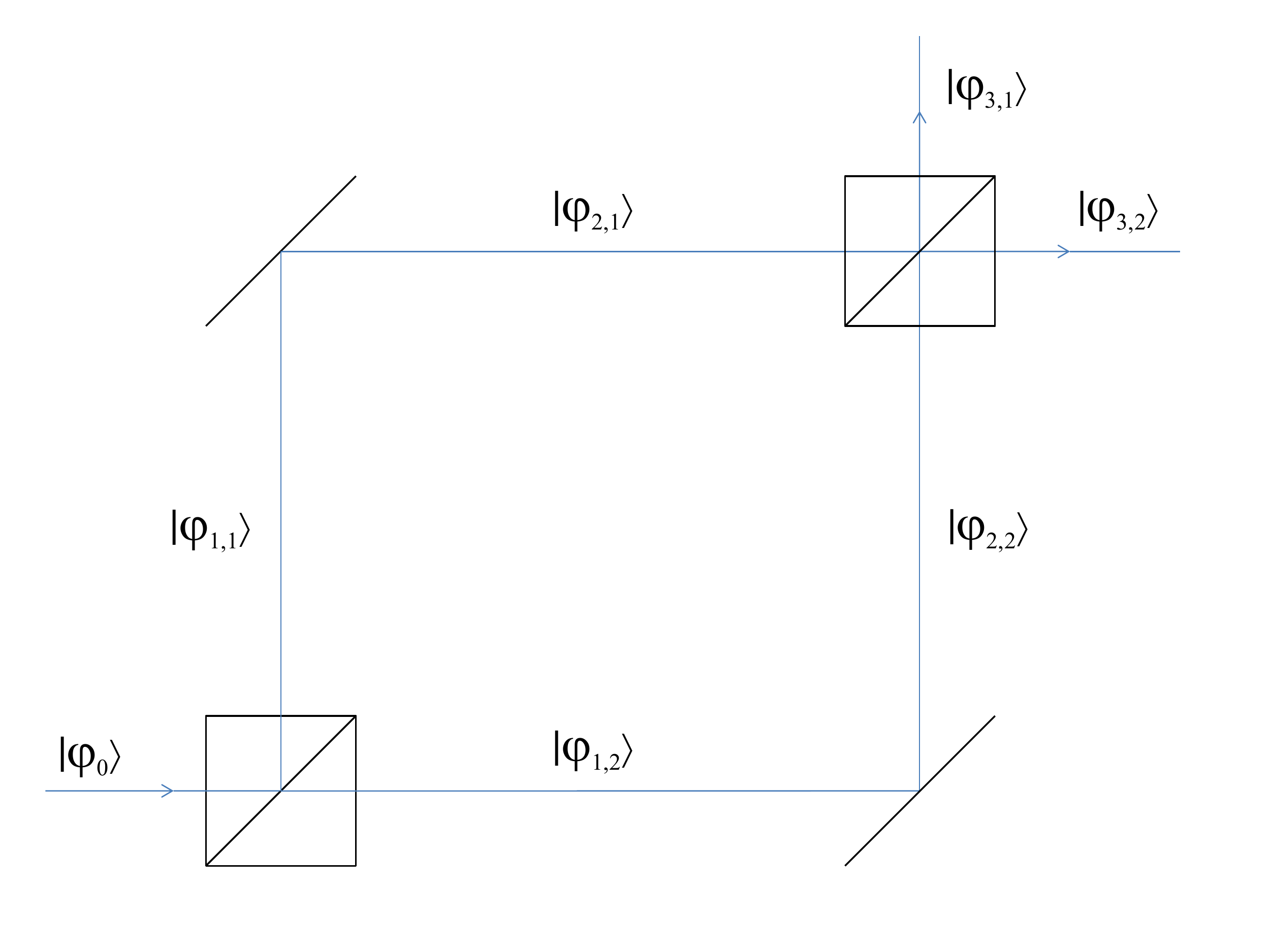}
 \caption[Fig2.]{The Mach-Zender interferometer with an input state $|\phi_{0}\rangle$ - a vacuum state is omitted which does not change further considerations. The beam-splitters can be represented by Hadamard operation
 $H=\frac{1}{\sqrt{2}} \left[
            \begin{array}{cc}
              1 & 1 \\
              1 & -1 \\
            \end{array}
          \right]$ acting on the spatial modes.
 One can analyze the interferometer via four-times histories on times $t_{0}<t_{1}<t_{2}<t_{3}$ for the interferometer process: $|\phi_{0}\rangle \rightarrow \frac{1}{\sqrt{2}}(|\phi_{1,1}\rangle+|\phi_{1,2}\rangle)\rightarrow \frac{1}{\sqrt{2}}(|\phi_{2,1}\rangle+|\phi_{2,2}\rangle)\rightarrow |\phi_{3,2}\rangle$.}
  \label{Fig 2}
\end{figure}

Besides the aforementioned reasoning derived from Feynman's quantum paths, one can refer to a broadly used explanation \cite{Wootters} for spatial monogamy of entanglement
between parties ABC (or further $\{t_{3}, t_{2}, t_{1}, t_{0}\}$ for temporal correlations). It states that A cannot be simultaneously fully entangled with B and C since then AB would be entangled with C having a mixed density
matrix that contradicts purity of the singled state shared between A and B. For the history spaces one can build naturally $\mathcal{C}^{*}$-Algebra of history operators equipped with a partial trace operation \ref{PartialTrace} and follow the same reasoning for entangled histories.
We can summary these considerations with the following corollary about monogamy of temporal entangled histories:

\begin{cor}
There does not exist  any such a history $|H)\in Proj(\mathcal{H}^{\otimes n})$ so that for three chosen times $\{t_{3},t_{2},t_{1}\}$ one can find reduced histories $|\Psi_{t_{3}t_{2}})=\frac{1}{\sqrt{2}}(|e_{0})\odot|e_{0})+|e_{1})\odot|e_{1}))$ and $|\Psi_{t_{2}t_{1}})=\frac{1}{\sqrt{2}}(|e_{0})\odot|e_{0})+|e_{1})\odot|e_{1}))$.
\end{cor}


This lemma holds for any finite dimension $n$ and also for general entangled states of the form \ref{maxent}.

As a consequence, there does not exist such a temporal observable $\widehat{\Lambda}_{A_{1}A_{2}A_{3}}$ so that $A_{1}A_{2}$ parties are maximally entangled and $A_{2}A_{3}$ are maximally entangled simultaneously on times $\{t_{3}, t_{2}, t_{1}\}$. However, in principle there exist
observables of different histories that do not commute and cannot be observed at the same reference frame by an observer that are maximally entangled between $A_{1}A_{2}$ and $A_{2}A_{3}$ \cite{MNPH3}.


It might be also interesting to consider a temporal analogy of spatially separable states. A natural consequence of entanglement monogamy in space is that all $\infty$-extendible states \cite{D2, T1, MN3}
are separable, i.e. we cannot build a quantum spatial state where a chosen party is entangled with an infinite number of parties. In principle, if we consider now Feynmann path integral which integrates all probability amplitudes over possible paths, one can state a question about correlations between a state of a system at a chosen time $t_{x}$ and all other times separated by $dt$ in this
evolution. Suppose we are considering a two-state history $|F)=[x_{E}]\odot [x_{S}]$ where a particle is localized at $x_{S}$ at time $t_{S}$ (our initial state is $|S\rangle=|x_{S}\rangle$) and evolves to the final state $|E\rangle$ localized at $x_{E}$ at time $t_{E}$. This history can be further expanded
as a Feynman path integral \cite{Schwartz} assuming breaking down the evolution time into $n$ small time intervals $\delta t$:
\begin{equation}
\begin{split}
\langle E|S\rangle = \int dx_{n}\cdots dx_{1} \langle x_{E}| e^{-iH(t_{E})\delta t}|x_{n}\rangle\langle x_{n}|\cdots \\
 |x_{2}\rangle\langle x_{2}| e^{-iH(t_{2})\delta t}|x_{1}\rangle\langle x_{1}|e^{-iH(t_{1})\delta t} |x_{S}\rangle
\end{split}
\end{equation}
where we assumed evolution steered by hamiltonian $H$ being a smooth function of $t$.
One can represent the product $|E\rangle\langle E|S\rangle\langle S|$ by means of integration over histories (it is crucial to remember that as in case of path integral summands not every
particular history summand has to be inherently consistent, i.e. physically realizable) as follows:
\begin{equation}
\begin{split}
|E\rangle\langle E|S\rangle\langle S|=\int dx_{n}\cdots dx_{1} K([x_{E}]\odot[x_{n}]\odot\ldots \\
\odot[x_{2}]\odot[x_{1}]\odot[x_{S}])
\end{split}
\end{equation}
It represents an expansion of a quantum propagator via quantum histories and
in general, the history $|F)=[x_{E}]\odot [x_{S}]$ is separable.

\section{Legett-Garg Inequalities and Tsirelson's Bound from Entangled Histories}

The violation of local realism (LR) \cite{Bell} and macrorealism (MR) \cite{MRealism} by quantum theories has been studied for many years in experimental setups where measurements' data are tested against violation of Bell inequalities for LR and Leggett-Garg inequalities (LGI) \cite{LGI} for MR. For quantum theories, the former raises as a consequence of non-classical correlations in space while the latter as a consequence of non-classicality of dynamic
evolution. In this section we show that entangled histories approach gives the same well-known Tsirelson bound \cite{Tsirelson} on quantum correlations for LGI as quantum entangled states in case of bi-partite spatial correlations for CHSH-inequalities.

In temporal version of CHSH-inequality being a modification of original Legett-Garg inequalities, Alice performs measurement at time $t_{1}$ choosing between two dichotomic observables $\{A_{1}^{(1)}, A_{2}^{(1)}\}$ and then Bob performs a measurement at time $t_{2}$ choosing between $\{B_{1}^{(2)}, B_{2}^{(2)}\}$. Therefore, the structure of this LGI can be represented as follows \cite{Vedral}:
\begin{equation}
S_{LGI}\equiv c_{12}+c_{21}+c_{11}-c_{22} \leq 2
\end{equation}
where $c_{ij}=\langle A_{i}^{(1)}, B_{j}^{(2)}  \rangle$ stands for the expectation value of consecutive measurements performed at time $t_{1}$ and $t_{2}$.

Since one can build in a natural way $\mathcal{C}^{*}$-Algebra of history operators for normalized histories from projective Hilbert spaces equipped with a well-defined inner product, we provide reasoning about bounding the LGI purely on the space of entangled histories and achieve the quantum bound $2\sqrt 2$ of CHSH-inequality specific for spatial correlations. The importance of this result achieved analytically is due to the fact that
previously it was derived basing on convex optimization methods by means of semi-definite programming \cite{Budroni} and by means of correlator spaces \cite{Fritz} not being equivalent to probability space (probability conditional distributions
of consecutive events) without underpinning mathematical structure of quantum temporal states.

We will now remind the theory by B.S. Cirel'son about bounds on Bell's inequalities that is used for finding quantum bounds on spatial Bell-inequalities:

\begin{thm}\cite{Tsirelson}
The following  conditions are equivalent for real numbers $c_{kl}$, $k=1,\ldots, m$, $l=1,\ldots, n$:
\\
\\
 $1.$ There exists $\mathcal{C}^{*}$-Algebra $\mathcal{A}$ with identity, Hermitian operators $A_{1},\ldots, A_{m}, B_{1},\ldots, B_{n} \in \mathcal{A}$ and a state $f$ on $\mathcal{A}$ so that for every ${k,l}$:\\
 \begin{equation}
 A_{k}B_{l}=B_{l}A_{k}; \; \mathbb{I}\leq A_{k}\leq \mathbb{I}; \; \mathbb{I}\leq B_{l}\leq \mathbb{I}; \; f(A_{k}B_{l})=c_{kl}.
\end{equation}
 $2. $ There exists a density matrix $W$ such that for every $k,l$:
 \begin{equation}
Tr(A_{k}B_{l}W)=c_{kl} \; and \; A_{k}^{2}=\mathbb{I};\; B_{l}^{2}=\mathbb{I}.
\end{equation}
$3. $ There are unit vectors $x_{1}, \ldots, x_{m}, y_{1},\ldots, y_{n}$ in a $(m+n)$-dimensional Euclidean space such that:
\begin{equation}
\langle x_{k},y_{l} \rangle = c_{kl}.
\end{equation}
\end{thm}

In a temporal setup one considers measurements $\mathbb{A}=I \odot \mathbb{A}^{(1)}$ (measurement $\mathbb{A}$ occurring at time $t_{1}$) and $\mathbb{B}=\mathbb{B}^{(2)}\odot I$ which is in an exact analogy to the proof of the above theorem for a spatial setup. The history with 'injected' measurements can be represented as $|\widetilde{H})=\alpha \mathbb{A}\mathbb{B}|H)\mathbb{A}^{\dagger}\mathbb{B}^{\dagger}$ where $\alpha$ stands for a normalization factor.
History observables are history state operators which are naturally Hermitian and their eigenvectors can generate a consistent history family\cite{WC1}. For an exemplary observable $A=\sum_{i}a_{i}|H_{i})(H_{i}|$, its measurement on a history $|H)$ generates an expectation value $\langle A\rangle=Tr(A|H)(H|)$ (i.e. the result $a_{i}$ is achieved with probability $|(H|H_{i})|^{2}$) in analogy to the spatial case. Thus, one achieves history $|\widetilde{H})$ as a realized
history with measurements and the expectation value of the history observable $\langle A \rangle$. It is worth mentioning that $|\widetilde{H})$ and $|H)$ are both compatible histories, i.e. related by a linear transformation. Equipped with the aforementioned findings about history observables, one can state now the following lemma:
\begin{lem}
For any history density matrix $W$ and Hermitian history dichotomic observables $A_{i}=I\odot A_{i}^{(1)}$ and $B_{j}= B_{j}^{(2)} \odot I$ where $i,j \in \{1,2\}$ the following bound holds:
\begin{eqnarray}
S_{LGI}&=&c_{11}+c_{12}+c_{21}-c_{22}\\
&=&Tr((A_{1}B_{1}+A_{1}B_{2}+A_{2}B_{1}-A_{2}B_{2})W)\nonumber \\
&\leq& 2\sqrt{2}
\nonumber
\end{eqnarray}
\end{lem}
\begin{proof}
The proof of this observation can be performed in similar to the spatial version of CHSH-Bell inequality under assumption that the states are represented by
entangled history states and for two possible measurements $\{A_{1}^{(1)}, A_{2}^{(1)}\}$ at time $t_{1}$ and two measurements $\{B_{1}^{(1)}, B_{2}^{(1)}\}$ at time $t_{2}$. These operators can be of dimension $2\times2$ meeting the condition $A_{i}^{2}=B_{j}^{2}=I$. Therefore, they can be interpreted as spin components along two different directions. In consequence, it is well-known that the above inequality is saturated for $2\sqrt{2}$ for a linear combination of tensor spin correlation that holds also for temporal correlations.
Additionally, one could also apply for this temporal inequality reasoning based on the following obvious finding \cite{Tsirelson} that holds also for the temporal scenario due to the structure of $\mathcal{C}^{*}$-Algebra of history operators with $\odot$-tensor operation:
\begin{eqnarray}
A_{1}B_{1}+A_{1}B_{2}+A_{2}B_{1}-A_{2}B_{2}&\leq&\\
\frac{1}{\sqrt{2}}(A_{1}^2+A_{2}^2+B_{1}^2+B_{2}^2)&\leq&
2\sqrt{2}I \nonumber
\end{eqnarray}

\end{proof}




\section{Conclusions}

In this paper we presented monogamous properties of quantum entangled histories proving that
quantum entanglement in time has properties similar to quantum entanglement in space. We pointed out that
a Tsirelson-like bound can be calculated for Leggett-Garg inequalities analytically applying entangled histories which
is a new result in comparison to the limits calculated numerically by means of semi-definite programming. We also introduced a partial trace operation for entangled histories which is operationally important for analysis of reduced histories.
There are many open problems and questions for further research in this field. Entangled histories approach is a substantial modification of the original consistent histories approach, especially in relation to the entanglement in time introducing non-locality of time into the framework. Future research can be focused on analysis of non-locality in time and finding more appropriate mathematical structures that will enable easier calculations of measurements' outputs for observers in different reference frames. Monogamy of entanglement in time and non-locality in time can be probably applied also in quantum cryptography and should give some new insights into non-sequential quantum algorithms and information processing. Finally, as stated in the paper the subject is fundamental for understanding relativistic quantum information theory and brings new prospects for this field.


\section{Acknowledgments}
The author would like to thank Robert Griffiths for answering many questions related to his theory. Acknowledgments to Pawel Horodecki for critical comments and discussions on this paper. This work was supported by by the ERC
grant QOLAPS. Part of this work was performed at the National Quantum Information Center in Gdansk.


\begin{thebibliography}{8}
\bibitem{Bell}
J. S. Bell, Physics {\bf 1} (3), 195 (1964).
\bibitem{Aspect}
A. Aspect, J. Dalibard, G. Roger, Phys. Rev. Lett. {\bf 49}, 1804 (1982).
\bibitem{Freedman}
S. J. Freedman, J. F. Clauser, Phys. Rev. Lett. {\bf 28}, 938 (1972).
\bibitem{EX1}
C. Robens et al., Phys. Rev. A 5, 011003 (2015).
\bibitem{EX2}
A. Asadin, C. Brukner and P. Rabl, Phys. Rev. Lett. {\bf 112}, 190402 (2014).
\bibitem{EX3}
H. Katiyar et al., Phys. Rev. A {\bf 87}, 052102 (2013).
\bibitem{EX4}
A. M. Souza, I. S. Oliveira and R. S. Sarthour, New J. Phys. 13, 053023 (2011).
\bibitem{LGI2}
A. J. Leggett, J. Phys.: Condens. Matter {\bf 14}, R415 (2002).
\bibitem{Lutk1}
G. O. Myhr, N. L\"{u}tkenhaus, Phys. Rev. A {\bf 79}, 062307 (2009).
\bibitem{Lutk2}
G. O. Myhr et. al., Symmetric extensions in two-way quantum key distribution,
Preprint quant-ph/0812.3607v1.
\bibitem{Lutk3}
T. Moroder, N. L\"{u}tkenhaus, Phys. Rev. A 74, 052301 (2006).
\bibitem{Lutk4}
Jianxin Chen et. al, Phys. Rev. A {\bf 90}, 032318 (2014).
\bibitem{Devetak05}
I. Devetak and A. Winter, Proc. R. Soc. Lond. A {\bf 461}, 207 (2005).
\bibitem{KLi}
K. Li, A. Winter, Squashed entanglement, k-extendibility, quantum Markov chains, and recovery maps, Preprint quant-ph/1410.4184 (2014).
\bibitem{MNPH}
M. L. Nowakowski, P. Horodecki,  J. Phys. A: Math. Theor. 42, 135306 (2009).
\bibitem{MNPH2}
M. L. Nowakowski, P. Horodecki,  Phys. Rev. A 82, 042342 (2010).
\bibitem{MNPH3}
M. Nowakowski, P. Horodecki, in preparation.
\bibitem{Vaidman}
Y. Aharonov, L. Vaidman, The two-state vector formalism of quantum mechanics, in Time in Quantum Mechanics, Springer, 369 (2002).
\bibitem{RG1}
R. Griffiths, Journal of Statistical Physics 36.1-2, 219-72 (1984).
\bibitem{RG2}
R. Griffiths,  Phys. Rev. Lett {\bf 70}, 2201-204 (1993).
\bibitem{RG3}
R. Griffiths, Consistent Quantum Theory, Cambridge: Cambridge UP, 2002.
\bibitem{RG4}
R. Griffiths, Consistent Quantum Measurements, Preprint quant-ph/1501.04813   (2015).
\bibitem{RG5}
R. Griffiths, Phys. Rev. A 54, 2759 (1996).
\bibitem{RG6}
R. Griffiths, R. Omn`es, Physics Today
{\bf 52}, 26-31 (1999).
\bibitem{RG7}
R. Griffiths, Private communication.
\bibitem{CJ1}
C. J. Isham,
Journal of Math. Phys. {\bf 35}, 2157 (1994).
\bibitem{CJ2}
C. J. Isham and N. Linden, Journal of Math.
Phys. {\bf 35}, 5452 (1994).

\bibitem{WC1}
J. Cotler, W. Wilczek, Entangled Histories, Preprint quant-ph/1502.02480 (2015).
\bibitem{WC2}
J. Cotler, W. Wilczek, Bell Tests for Histories, Preprint quant-ph/1503.06458  (2015).

\bibitem{Wootters}
Valerie Coffman, Joydip Kundu, William K. Wootters, Phys. Rev. A {\bf 61}, 052306 (2000).

\bibitem{Schwartz}
M. D. Schwartz, Quantum Field Theory and the Standard Model, Cambridge, 2013.

\bibitem{Feynman}
R. P. Feynman, Space-time approach to non-relativistic quantum mechanics, Rev. Mod. Phys.
{\bf 20}, 367 (1948).
\bibitem{D2}
A. C. Doherty, P. A. Parillo and F. M. Spedalieri, Phys. Rev. A
{\bf 69}, 022308 (2004).
\bibitem{T1}
B. M. Terhal, A. C. Doherty and D. Schwab, Phys. Rev. Lett. {\bf
90}, (2003).
\bibitem{MN3}
M. Nowakowski, Symmetric extendiblity of quantum states, Preprint quant-ph/1504.00388 (2015).
\bibitem{Tsirelson}
B. S. Cirel'son, Lett. Math. Phys. {\bf 4}, 93-100 (1980).
\bibitem{Bell}
J. S. Bell, Physics (New York) 1, 195 (1964).
\bibitem{LGI}
A. J. Leggett and A. Garg, Phys. Rev. Lett. 54, 857 (1985).
\bibitem{MRealism}
A. J. Leggett, J. Phys.: Cond. Mat. 14, R415 (2002).
\bibitem{Vedral}
C. Brukner, S. Taylor, S. Cheung, V. Vedral, Quantum Entanglement in Time, Preprint quant-ph/0402127, (2004).
\bibitem{Fritz}
T. Fritz, New J. Phys. 12, 083055 (2010).
\bibitem{Budroni}
C. Budroni, T. Moroder, M. Kleinmann, O. Gühne, Phys. Rev. Lett. 111, 020403 (2013).






%

\end{thebibliography}
\end{document}